\documentclass[conference]{IEEEtran}
\IEEEoverridecommandlockouts
\usepackage{cite}
\usepackage{amsmath,amssymb,amsfonts}
\usepackage{algorithmic}
\usepackage{graphicx}
\usepackage{textcomp}
\usepackage{xcolor}
\usepackage{framed} 
\usepackage{stfloats}
\usepackage{threeparttable}
\usepackage{extarrows}
\usepackage{hyperref}

\usepackage{amsthm} 
\newtheorem{theorem}{Theorem}

\newtheorem{lemma}{Lemma}
\newtheorem{definition}{Definition}

\def\BibTeX{{\rm B\kern-.05em{\sc i\kern-.025em b}\kern-.08em
    T\kern-.1667em\lower.7ex\hbox{E}\kern-.125emX}}
\begin{document}

\title{A Secure Remote Password Protocol From The Learning With Errors Problem\\
}

\author{\IEEEauthorblockN{1\textsuperscript{st} Huapeng Li}
\IEEEauthorblockA{\textit{College of Information} \\
\textit{North China University of Technology}\\
Beijing, China \\
bwvanh114@outlook.com}
\and
\IEEEauthorblockN{2\textsuperscript{nd} Baocheng Wang}
\IEEEauthorblockA{\textit{College of Information} \\
\textit{North China University of Technology}\\
Beijing, China \\
wbaocheng@ncut.edu.cn}
}

\maketitle

\begin{abstract}
Secure Remote Password (\textsf{\textmd{SRP}}) protocol is an important password-authenticated key exchange (\textsf{\textmd{PAKE}}) protocol based on the discrete logarithm problem (\textsf{\textmd{DLP}}). The protocol is specifically designed to work for parties to obtain a session key and it is widely used in various scenarios due to its attractive security features. As an augmented \textsf{\textmd{PAKE}} protocol, the server does not store password-equivalent data. And this makes attackers who manage to steal the server data cannot masquerade as the client unless performing a brute force search for the password. However, the advance in quantum computing has potentially made classic \textsf{\textmd{DLP}}-based public-key cryptography schemes not secure, including \textsf{\textmd{SRP}} protocol. So it is significant to design a new protocol resisting quantum attacks. In this paper, based on the basic protocol, we build a post-quantum \textsf{\textmd{SRP}} protocol from the learning with errors (\textsf{\textmd{LWE}}) problem. Besides the resistance to known quantum attacks, it maintains the various secure qualities of the original protocol.
\end{abstract}

\renewcommand{\thefootnote}{}
\footnotetext{This paper is a corrected version of the previously published paper (\href{https://ieeexplore.ieee.org/document/10063693}{DOI: 10.1109/TrustCom56396.2022.00149}). The changes primarily involve: rectifying some symbol errors in the original sections, supplementing the missing references with corresponding descriptions, and correcting some formatting and spelling mistakes.}


\section{Introduction}
Wu TD. proposed the \textsf{\textmd{Secure Remote Password}} (\textsf{\textmd{SRP}}) protocol \cite{srp0} in 1998. The \textsf{\textmd{SRP}} protocol is an augmented password-authenticated key exchange (\textsf{\textmd{PAKE}}) protocol. In \textsf{\textmd{PAKE}} protocols, an eavesdropper cannot obtain enough information to get the password by brute force or a dictionary attack unless further information can be obtained from the corrupted parties. Clients or servers using \textsf{\textmd{SRP}} do not transmit passwords in clear text or encrypted form over the network, thus it eliminates password spoofing completely and that is why the password cannot be derived by an eavesdropper. On the other hand, the server needs to know the necessary information about the password (but not the password itself, e.g., a “verifier”) to build a secure connection, which guarantee the validity of identities. In general, using \textsf{\textmd{SRP}} has the following benefits for both parties: 1). \textsf{\textmd{SRP}} resists “Password Sniffing” attacks. In a session using \textsf{\textmd{SRP}} authentication, the listener will not be able to monitor any passwords transmitted over the network; 2). \textsf{\textmd{SRP}} resists dictionary attacks. In the process of generating a verifier, the protocol has already carried out password security processing. The algorithm used requires the attacker to perform an impossibly computing before a powerful attack can be made. In this case, even if an attacker tracks the entire session and compares the entire message to a regular password in a dictionary, he cannot obtain any available information associated with the password; 3). \textsf{\textmd{SRP}} resists attacks on passwords. Even if a very weak password is active, it is still too hard for an intruder to crack; 4). \textsf{\textmd{SRP}} is simple to implement because there is no key distribution center or authentication server needed in the protocol. Due to its security and ease of deployment, the \textsf{\textmd{SRP}} protocol can also be used in scenarios requiring identity authentication in cloud computing, mobile network, Internet of Things and other technology domains (e.g., \cite{z1zhao2021novel}, \cite{z2peng2021secure}).

However, at present, the existing key exchange protocols (including \textsf{\textmd{SRP}}) are generally based on the discrete logarithm problem (\textsf{\textmd{DLP}}), the integer factorization problem (\textsf{\textmd{IFP}}), and so on. These difficult problems were once considered intractable. However, the design of quantum computers and the proposal of related algorithms (e.g., Shor’s algorithm \cite{shor1999polynomial}) break the security of the schemes constructed based on the above problems. Therefore, designing cryptographic structures that can resist quantum attacks has become a focus. Fortunately, several classes of quantum-resistant cryptosystems have been devised \cite{Introduction}. Among them, the lattice-based cryptosystem is flexible in design, high in computational efficiency, and it can construct almost any primitive in cryptography. The first lattice-based cryptosystem is the Ajtai-Dwork cryptosystem proposed in 1997 \cite{Apublic-key}. And the security of the system is based on Ajtai’s average-case to worst-case reduction and it is proved the shortest vector problem (\textsf{\textmd{SVP}}) is \textsf{\textmd{NP}}-hard \cite{svp}. Since then, lattice-based cryptography thrives, a series of difficult problems on lattices have been discovered, and a large number of cryptographic structures based on these problems have been proposed \cite{Adecadeof-decade}. Among many lattice-based cryptographic constructions, the learning with errors (\textsf{\textmd{LWE}}) problem introduces “errors” (or “noise”) as random disturbances that mask real information and is an extremely flexible construction idea. The \textsf{\textmd{LWE}} problem can be reduced to the closest vector problem (\textsf{\textmd{CVP}}) \cite{cvp}. The \textsf{\textmd{CVP}} problem can be reduced to the \textsf{\textmd{SVP}} problem \cite{Approximatingcsvphard} and there exists no efficient quantum computing algorithm to solve it. The module \textsf{\textmd{LWE}} problem (\textsf{\textmd{MLWE}}), the ring \textsf{\textmd{LWE}} problem (\textsf{\textmd{RLWE}}), and other \textsf{\textmd{LWE}}-like problems \cite{mlwe-Algebraically} were subsequently proposed. Unlike the \textsf{\textmd{LWE}} problem, these problems require special algebraic structures and the security of the cryptosystems based on these needs more considerations.

Therefore, to construct a simple quantum-resistant \textsf{\textmd{SRP}} protocol that is not dependent on a specific algebraic structure is the motivation of this paper. To solve this problem, we first design an extended version - the multi-instance \textsf{\textmd{LWE}} problem. And based on the original \textsf{\textmd{SRP}} process \cite{srp0}, we generalize Ding’s key exchange scheme \cite{ding-Asimpleprovably}  to the extended multi-instance \textsf{\textmd{LWE}} structure, building an \textsf{\textmd{LWE}}-based \textsf{\textmd{SRP}} scheme inspired by the design thoughts of quantum secure \textsf{\textmd{SRP}} based on \textsf{\textmd{RLWE}} difficulty problem \cite{rlwesrp}. The scheme only relies on the standard \textsf{\textmd{LWE}} problem to achieve a simple and generic problem structure. In addition, we give a formal proof that is not considered in the original \textsf{\textmd{SRP}} scheme. The protocol maintains various security features of the original design and resists known quantum computing attacks due to the backbone structure based on the hard lattice problem.

\subsection{Organization}
The organization of the remainder of this paper is organized as follows: In Section \ref{Preliminaries} 
we present notations and  definitions used throughout the paper; In Section \ref{LWE PKE} we review 
\textsf{\textmd{LWE}}-based encryption schemes briefly; In Section \ref{SRP} we review 
the original \textsf{\textmd{SRP}} protocol; In Section \ref{LWE PKE} we describe our \textsf{\textmd{LWE}}-based 
\textsf{\textmd{SRP}} scheme; Finally, in Section \ref{Conclusions} we give conclusions of the paper.

\section{Preliminaries}\label{Preliminaries}

In this section, we present related notations, definitions, and theorems. Some of them are from previous works.

\subsection{Lattice and the Learning with Errors Problem}

\begin{definition}{\rm (Lattice).}
	Given a set of linearly independent vectors in $\mathbb{R}^n$: $\{ \boldsymbol{b}_1, \boldsymbol{b}_2, \ldots, \boldsymbol{b}_m \} \in \mathbb{R}^{n \times m}$, a lattice $L(\mathbf{B}) \in \mathbb{R}^n$ spawned by $\{ \boldsymbol{b}_1, \boldsymbol{b}_2, \ldots, \boldsymbol{b}_m \}$ is defined as:
	\begin{center}
    \begin{equation}
        	\mathcal{L}(\mathbf{B})=\left\{\sum\limits_{i=1}^{m} \boldsymbol{b}_{i} x_{i}: x_{i} \in \mathbb{Z}\right\},
    \end{equation}
	
	\end{center}
	where $n$ and $m$ are integers denoting the dimension and rank of $L(\boldsymbol{B})$. And the matrix $\boldsymbol{B} = \{ \boldsymbol{b}_1, \boldsymbol{b}_2, \ldots, \boldsymbol{b}_m \}$ is the basis of the generated lattice.
\end{definition}


%

\begin{definition}{\rm (Normal distribution\cite{regev2009lattices-Onlattices}).}
	The normal distribution with mean $\mu = 0$, variance denoted by $\sigma^2$ is a probability
distribution on $\mathbb{R}$ and the probability density function is as follows:

	\begin{center}
    \begin{equation}
       	\frac{1}{\sqrt{2 \pi} \cdot \sigma} \exp \left(-\frac{1}{2}\left(\frac{x}{\sigma}\right)^{2}\right).
    \end{equation}

	\end{center}

Let $\rho_s(x) = \exp\left(-\pi \frac{\|x\|^2}{s^2}\right)$ be a normal distribution function with the parameter $s > 0$ and $x$ denotes a vector. Then we have an $n$-dimensional joint normal distribution function:
\begin{center}
\begin{equation}
      \int_{x \in \mathbb{R}^n} \rho_s(x) \, dx := \frac{\rho_s}{s^n}.
\end{equation}
  
\end{center}
For a discrete set \( S \), any vector \( v \in \mathbb{R}^n \), use \( \rho_{s,v}(x) := \rho_s(x - v) \) to analogy the form of \( \rho_s \). Then, we get the joint discrete normal probability distribution \( D_{S,s} \):

	\begin{center}
    \begin{equation}
        \forall \boldsymbol{x} \in S, \mathcal{D}_{S, s}(\boldsymbol{x})=\frac{\rho_{s}(\boldsymbol{x})}{\rho_{s}(S)}.
    \end{equation}
		
	\end{center}
\end{definition}

\begin{definition}{\rm (Standard \textsf{LWE} distribution \cite{regev2009lattices-Onlattices}).}

Generate two integers \( n \) and \( m \) with a security parameter \( \kappa \) as the seed: \( n = n(\kappa) \) denotes the amount of equations, \( m = m(\kappa) \) denotes the dimension. Sample noise according to the distribution \( \chi = \chi(\kappa) \) over \( \mathbb{Z} \) with the bound \( B = B(\kappa) \), \( q = q(\kappa) \geq 2 \) is an integer modulus for any polynomial \( p = p(\kappa) \) with the lower bound \( q \geq 2pB \). Sample a secret vector \( s \in \mathbb{Z}_q^{n \times 1} \), then an {\rm \textsf{LWE}} \textit{distribution\ $\mathcal{A}_{\boldsymbol{s}, \chi}  \ over  \ \mathbb{Z}_{q}^{n} \times \mathbb{Z}_{q}$ is sampled by choosing $\mathbf{A} \in \mathbb{Z}_{q}^{m \times n}$  uniformly at random, choosing  $\boldsymbol{e} \leftarrow \chi^{m \times 1}$,} and then get an LWE instance: \(\mathbf{A}, \boldsymbol{b}=\mathbf{A}\boldsymbol{s}+\boldsymbol{e}(\bmod\ q) \).

\end{definition}

\begin{definition}{\rm (Standard decisional-\textsf{LWE}$\operatorname{}_{{n \times m},{m \times 1}, q, \chi}$\cite{regev2009lattices-Onlattices}).}
	Given an independent sample $(\mathbf{A}, \boldsymbol{b}) \in \mathbb{Z}_{q}^{n \times m} \times \mathbb{Z}_{q}^{m \times 1}$ , where the sample is distributed according to either: \textbf{1)}.  $\mathcal{A}_{s, \chi}$  for a uniformly random 
	$\boldsymbol{s} \in \mathbb{Z}_{q}^{{n}}=\{(\mathbf{A}, \boldsymbol{b}): \mathbf{A} \leftarrow \mathbb{Z}_{q}^{n \times m}, \boldsymbol{s} \leftarrow \mathbb{Z}_{q}^{m}, \boldsymbol{e} \leftarrow 
	\chi^{n}, \boldsymbol{b}=\mathbf{A} \boldsymbol{s}+\boldsymbol{e}\ (\bmod\  q)\}$), or \textbf{2)}. $\mathcal{U}$, the uniform distribution. And the results of the two sampling methods are computationally indistinguishable.
\end{definition}

In this paper, we design the protocol from an extended multi-instance \textsf{LWE} problem. And with the definition below, we need to take operations from both sides of the space \( \mathbf{A} \) instead of operations on elements in a commutative ring \cite{ding-Asimpleprovably} \cite{rlwesrp}. And this can also be seen as a general construction in the multi-instance \textsf{LWE} problem.

\begin{definition}{\rm (Extended multi-instance decisional-\textsf{LWE}$\operatorname{}_{{n \times n},{n \times n}, q, \chi}$).}
	Given an independent sample $(\mathbf{A}, \mathbf{B}) \in \mathbb{Z}_{q}^{n \times n} \times \mathbb{Z}_{q}^{n \times n}$ , where the sample is distributed according to either: a).  $\mathcal{A}_{\mathbf{S}, \chi}$  for a uniformly random 
	$\mathbf{S} \in \mathbb{Z}_{q}^{{n \times n}}=\{(\mathbf{A}, \mathbf{B}): \mathbf{A} \leftarrow \mathbb{Z}_{q}^{n \times n}, \mathbf{S} \leftarrow \mathcal{D}_{\mathbb{Z}_{q}^{n \times n}, \tau}, \mathbf{E} \leftarrow 
	\mathcal{D}_{\mathbb{Z}_{q}^{n \times n}, \tau}, \mathbf{B}=\mathbf{A}\mathbf{S}+\mathbf{E}\ (\bmod\ q)\}$), or b). $\mathcal{U}$, the uniform distribution. And the results of different sampling methods are also computationally indistinguishable.
\end{definition}

\subsection{Key Exchange}

The idea of using the \textsf{SRP} protocol to obtain the session key is similar to that of the Diffie-Hellman (\textsf{DH}) protocol. Here we briefly introduce the \textsf{DH} problem and the \textsf{DH} key exchange process.

\begin{definition}{\rm (\textsf{DH} protocol \cite{dh-Newdirections})}
		There are mainly four parameters in the protocol:
		$p$, $g$, $a$ and $b$, where $p$, $g$ are publicly available numbers and $a$, $b$ are private.  Specifically, $p$ is a prime, $g$ is a primitive root of $p$, $a$, and $b$ are private values decided by the parties respectively during the protocol. 
		
		Parties ($Alice$ and $Bob$) pick private values $a$ and $b$ and they generate key materials to send with each other publicly. The opposite person receives the key material and then generates a session key.
\end{definition}
\begin{figure}[h]
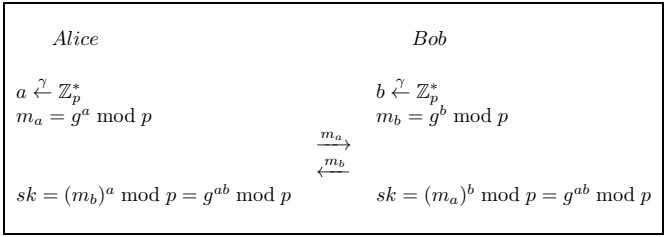

	\begin{center}
	\resizebox{1\columnwidth}{!}{
		\begin{threeparttable}
			\begin{tabular}{|lcl|}
				\hline&&\\
				~~~~~$Alice$& &~~~~~$Bob$\\
				&&\\
				
				$a \stackrel{\gamma}{\leftarrow}\mathbb{Z}_{p}^{*}$
				&&$b \stackrel{\gamma}{\leftarrow}\mathbb{Z}_{p}^{*}$\\
				
				$m_{a}=g^{a} \bmod p$
				&&$m_{b}=g^{b} \bmod p$\\
				
				&$\xrightarrow{m_{a}}$&\\
				&$\xleftarrow{m_{b}}$&\\
				
				$sk=(m_{b})^{a} \bmod p=g^{a b} \bmod p$
				&&$sk=(m_{a})^{b} \bmod p=g^{a b} \bmod p$\\
				
				&&\\
				\hline
			\end{tabular}
	\end{threeparttable}}
\end{center}
\caption{Diffie-Hellman key exchange.}
\end{figure}

\subsection{Basic Tools}
\begin{definition}{\rm (Signal function \cite{ding-Asimpleprovably}).}\label{signal}
	To reduce the noise added in the generation of key material, here use the technology of signal function from [12]. First, hint functions $\sigma_0(x)$, $\sigma_1(x)$ transfer a value from the space $\mathbb{Z}_q$ to the bit space $\{0, 1\}$:
	\begin{itemize}
		\item $\sigma_{0}(x)=\left\{\begin{array}{l}
			0, x \in\left[-\left\lfloor\frac{q}{4}\right\rfloor,\left\lfloor\frac{q}{4}\right\rfloor\right] \\
			1,otherwise
		\end{array}\right.,$\\
		\item $\sigma_{1}(x)=\left\{\begin{array}{l}
			0, x \in\left[-\left\lfloor\frac{q}{4}\right\rfloor+1,\left\lfloor\frac{q}{4}\right\rfloor+1\right] \\
			1,otherwise
		\end{array}\right.$
	\end{itemize}
Based on hint functions, let $\sigma(x)=\sigma_{b}(x)$ be the signal function, where $x \in \mathbb{Z}_{q}$,  $b \leftarrow\{0,1\}$.

\end{definition}

\begin{definition}{\rm (Extractor function).}
	To remove error perturbations in key material, we construct the function for the extended multi-instance  {\rm \textsf{LWE}} structure based on the  tool functions in \cite{ding-Asimpleprovably}. Let
	 $\varphi(\cdot)$ be a deterministic function: 
	 \begin{center}
	 $\varphi(x, \sigma)=(x+\sigma \cdot \frac{q-1}{2} \bmod q) \bmod 2$.	
	 \end{center}
	  The inputs of $\varphi(\cdot)$ are: $x \in \left\{\mu_{i,j}:[\mu_{i,j}]_{n \times n} \in \mathbb{Z}^{n \times n}\right\}$ and  the signal $\sigma$. The output of $\varphi(\cdot)$ is a bit matrix $[k_{i, j}]_{n \times n}$. For any entry of the matrix $x, y \in \mathbb{Z}_{q}$, if $\|x-y\|_{\infty} \leq \delta$, $\varphi(x, \sigma)=\varphi(y, \sigma)$, where $\sigma=\sigma(y)$, $\delta$ is the error tolerance. 
\end{definition}

\section{Public-key encryption From \textsf{\textmd{LWE}}}\label{LWE PKE}

In this section, we review the \textsf{LWE}-based public key cryptosystem and the indistinguishability of the decisional-\textsf{LWE} problem. The complete proof of correctness and security is omitted due to the limited space. (Please refer to \cite{regev2009lattices-Onlattices} and \cite{lwe-ana-Robustness} for complete proof and analysis.)
\begin{definition}{\rm (\textsf{LWE}-based public key cryptography \cite{regev2009lattices-Onlattices}).}
	The original \textsf{LWE}-based public key system is parameterized with the security parameter $n$, denoting the dimension of the space. Specifically, there are two integers $m$ and $p$ derived from $n$. $m$ denotes the scale of the equations and $p$ is a prime. And a probability distribution $\chi$ on $\mathbb{Z}_p$ is required to generate an error. To guarantee the security of the scheme, parameters hold: $m = 5n$, the prime $p > 2$ is required between $n^2$ and $2n^2$. The probability distribution $\chi$ is taken to be $\Psi_{\alpha(n)}$ for $\alpha(n) = o\left(\frac{1}{\sqrt{n} \log n}\right)$, i.e., $\alpha(n)$ is such that $\lim_{n \to \infty} \frac{\alpha(n)}{\sqrt{n} \log n} = 0$.
	
	\begin{itemize}
		\item {\rm \textsf{Private Key:}} Choose a secret $s \in \mathbb{Z}_p^n$ randomly to be the private key
		\item {\rm \textsf{Public Key:}} Choose $m$ vectors $\boldsymbol{a_1}, \ldots, \boldsymbol{a_m} \in \mathbb{Z}_p^n$ independently from the uniform distribution, choose $m$ elements $e_1, \ldots, e_m \in \mathbb{Z}_p$ independently with the distribution $\chi$. Then compute $b_i = \langle a_i, s \rangle + e_i$. And the public key is the generated \textsf{LWE} instance: $(a_i, b_i)_{i=1}^m$
		\item {\rm \textsf{Encryption:}} Encryption cases are $\left( \sum_{i \in S} \boldsymbol{a_i}, \sum_{i \in S} b_i \right)$ when the bit is 0 and $\left( \sum_{i \in S} \boldsymbol{a_i}, \left\lfloor\frac{p}{2}\right\rfloor + \sum_{i \in S} b_i \right)$ via adding a distinguished value when it is 1, where $S$ is a randomly chosen subset from the lattice space.
		\item {\rm \textsf{Decryption:}} 
        To get a plain text bit from the cipher, the distance to 0 is required to be computed. Consider a pair $(a, b)$, we get 0 if $b - \langle a, s \rangle$ is closer to 0 under the bound $\left\lfloor\frac{p}{2}\right\rfloor$ modulo $p$ and we get 1 if it is more than $\left\lfloor\frac{p}{2}\right\rfloor$.
        
	\end{itemize}
\end{definition}

\begin{theorem}{{\rm (Indistinguishability of the dicisional-\textsf{LWE} problem \cite{lwe-ana-Robustness})}} For a given secure parameter $\kappa \log q$, hash function set denoted as \textsf{H} $: \{0, 1\}^n \to \{0, 1\}^*$ are difficult to invert with the hardness $2^{-k}$. And the indistinguishability is described as:
	 For any super-polynomial $q=q(n)$, any $m=$ $poly(n)$, any $\beta, \gamma \in(0, q)$ such that $\gamma / \beta=$ $negl(n)$
	 
	 \begin{center}
     \begin{equation}
         	 	(\mathbf{A}, \mathbf{A} \mathbf{s}+\boldsymbol{x}, h(\boldsymbol{s})) \approx_{c}(\mathbf{A}, \boldsymbol{u}, h(\boldsymbol{s})),
     \end{equation}

	 \end{center}
 where $\mathbf{A} \leftarrow \mathbb{Z}_{q}^{m \times n}$, $\boldsymbol{s} \leftarrow \mathbb{Z}_{q}^{n}$ and $\boldsymbol{u} \leftarrow \mathbb{Z}_{q}^{m}$ are uniformly random and $\boldsymbol{x} \leftarrow \bar{\Psi}_{\beta}^{m}$.
 \label{th1}
\end{theorem}

\section{Original Secure Remote Password Protocol}\label{SRP}
In this section, we briefly introduce the process of the \textsf{SRP} protocol.
\subsubsection{Notions}

Notations used in the original SRP protocol are as follows:

\begin{figure}[h]
	\begin{center}
		\resizebox{1\columnwidth}{!}{
			\begin{threeparttable}
				\begin{tabular}{|cl|}
					\hline&\\
					$n$& A prime number and the operations are all based on this modulo. $n$.\\
					$g$& A generator of the group.\\
					$u$& A public parameter chosen randomly.\\
					$P$& The customized password from the user.\\
					$s$& Salt: A string chosen randomly.\\
					$x$& A private key generated from the key material during the protocol.\\
					$v$& Key part: The password verifier.\\
					$a, b$& Temporary private material.\\
					$A, B$ & Public keys.\\
					$\textsf{H}()$& Hash function.\\
					$m, n$& Concatenated strings.\\
					$K$& Session key.\\
					&\\
					\hline
				\end{tabular}
		\end{threeparttable}}
	\end{center}
	\caption{Notions in the \textsf{SRP} protocol.}
	\label{srp1}
\end{figure}

\subsubsection{{\rm\textsf{SRP}} protocol {\rm \cite{srp0}}}

In the initial phrase, the verifier is required to be generated first. Specifically, \textit{Alice} computes $x = \textsf{H}(s, P)$ and $v = g^x$, where $s$ is the salt generated randomly. Then, \textit{Bob} stores $v$ and $s$ as the identity information of \textit{Alice}. Finally, temporary $x$ is destroyed for security. The data that \textit{Bob} needs to keep is the secret verifier $v$ from \textit{Alice} and $v$ is also the key part in authentication process. Therefore, there is no need for \textit{Alice} to keep the public key of \textit{Bob}. The procedure of the protocol is described in Figure 3 and the details are as follows:

\begin{itemize}
	\item[1).] $Alice$ sends her identity to $Bob$.
	\item[2).] \textit{Bob} looks up \textit{Alice}'s information tuple and extracts the verifier and salt. Then, \textit{Bob} sends the salt to \textit{Alice} to help her rebuild the verifier. In this case, \textit{Alice} cannot generate $v$ unless she holds the password $P$ set at the beginning of the protocol. With this process, the authentication is completed.
	\item[3).] \textit{Alice} randomly chooses a random number $a$ and computes the key material $A = g^a$, then she sends the result to \textit{Bob}.
	\item[4).] \textit{Bob} randomly chooses a number $b$ and computes the key material $B = v + g^b$, then he sends the result to \textit{Alice}.
	\item[5).] \textit{Alice} and \textit{Bob} compute the shared value $S = g^{ab + bux}$ with the key material sent to each other as inputs. And the final equivalent value can be obtained only if \textit{Alice} recovers the verifier with the correct password $P$ in Step 2.
	\item[6).] Both parties compute the hash value using $S$ to generate the session key.
	\item[7).] \textit{Alice} sends $M_1$ to \textit{Bob}. With this, \textit{Bob} knows that \textit{Alice} has generated the session key successfully. And then, he computes $M_1$ to verify the correctness and validity.
	\item[8).] \textit{Alice} sends $M_1$ to \textit{Bob}. With this, \textit{Bob} knows that \textit{Alice} has generated the session key successfully. And then, he computes $M_1$ to verify the correctness and validity.
\end{itemize}

\begin{figure}[h]
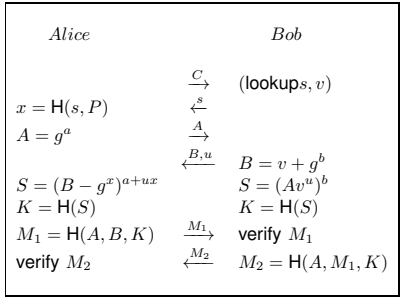

	\begin{center}
		\resizebox{0.6\columnwidth}{!}{
			\begin{threeparttable}
				\begin{tabular}{|lcl|}
					\hline&&\\
					~~~~~$Alice$& &~~~~~$Bob$\\
					&&\\
					
					&$\xrightarrow{C}$
					& $(\textsf{lookup}s, v)$\\
					
					$x=\textsf{H}(s,P)$
					&$\xleftarrow{s}$&\\
					
					$A=g^{a}$
					&$\xrightarrow{A}$&\\
					
					&$\xleftarrow{B, u}$
					&$B=v+g^{b}$\\
					
					$S=(B-g^{x})^{a+ux}$
					&&$S=(Av^{u})^{b}$\\
					
					$K=\textsf{H}(S)$
					&&$K=\textsf{H}(S)$\\
					
					$M_{1}=\textsf{H}(A, B, K)$
					&$\xrightarrow{M_{1}}$
					&$\textsf{verify}\ M_{1}$\\
					
					$\textsf{verify}\ M_{2}$
					&$\xleftarrow{M_{2}}$
					&$M_{2}=\textsf{H}(A, M_{1}, K)$\\

					&&\\
					\hline
				\end{tabular}
		\end{threeparttable}}
	\end{center}
	\caption{The original \textsf{SRP} protocol.}
	\label{srp2}
\end{figure}

\section{Secure Remote Password Protocol From LWE}\label{LWE SRP}

In this section, we describe the processes of the \textmd{LWE}-based
\textmd{SRP} protocol in detail.

\subsection{Design of {\rm \textsf{\textmd{LWE}}}-based {\rm \textsf{\textmd{SRP}}}}\label{design}
The complete process of the \textsf{LWE}-based \textsf{SRP} protocol is given in this section and Figure. \ref{lwesrp} shows a precise description.

\begin{itemize}
	\item[1).] $Client$ generates a lattice basis of the \textsf{LWE} problem through the public security parameter $\lambda$. Then he sets the password and randomly generates the $salt$, and combines these materials with $id$ to generate a random seed $\gamma$ (e.g., the hash value of these parameters).
	\item[2).] $Client$ generates the initial secret $\mathbf{S}_{I}$ and noise $\mathbf{E}_{I}$ according to the random seed $\gamma$, and constructs an \textsf{LWE} problem instance $\mathbf{V}$ as the verifier. Then the tuple$<id,salt,\mathbf{V}>$ is sent to $Server$.
	\item[3).] $Server$ keeps the message received from $Client$ and builds an index for $Client$ according to $id$. After sending the message successfully, $Client$ must completely erase all the values generated in the above processes for security. Certainly, $Client$ needs to keep the password. Up to now, the generation process of the verifier $\mathbf{V}$ has been completed and the next processes are for generating a session key.
	\item[4).] At the beginning of a key exchange process, both parties need to generate the lattice basis according to the public security parameter $\lambda$ as in the first step. Then each randomly generates $\mathbf{B}_{C}$ and $\mathbf{B}_{S}$ (the \textsf{LWE} problem instances) on this basis. Specifically, $Server$ needs to superimpose the verifier $\mathbf{V}$ in the exchange material $\mathbf{B}_{S}$ according to $id$ of $Client$ saved before. And $\mathbf{B}_{C}$ is sent to $Server$.
	\item[5).] $Server$ calculates the corresponding signal $\sigma$ based on the generated key material $\mathbf{M}_{S} =(\mathbf{V}+\mathbf{B}_{C})\mathbf{B}_{S}+2\mathbf{E}_{S}^{\prime}$. Then the tuple $<salt, \mathbf{B}_{S}, \sigma>$ is sent to $Client$.
	\item[6).] $Client$ first recovers the random seed $\gamma$ and then generates the secret $\mathbf{B}_{I}$ and the noise $\mathbf{E}_{I}$ in the verifier $\mathbf{V}$. $Client$ calculates his key material $\mathbf{M}_{C} =(\mathbf{S}_{I}+\mathbf{S}_{C})(\mathbf{B}_{S}-\mathbf{V})+2\mathbf{E}_{C}^{\prime}$.
	\item[7).] Both parties calculate the key materials respectively via the function $\varphi$ to obtain the equivalent value. Finally, the same session key is obtained via the key derivation function (\textsf{KDF}). 
\end{itemize}

\begin{figure}[h]
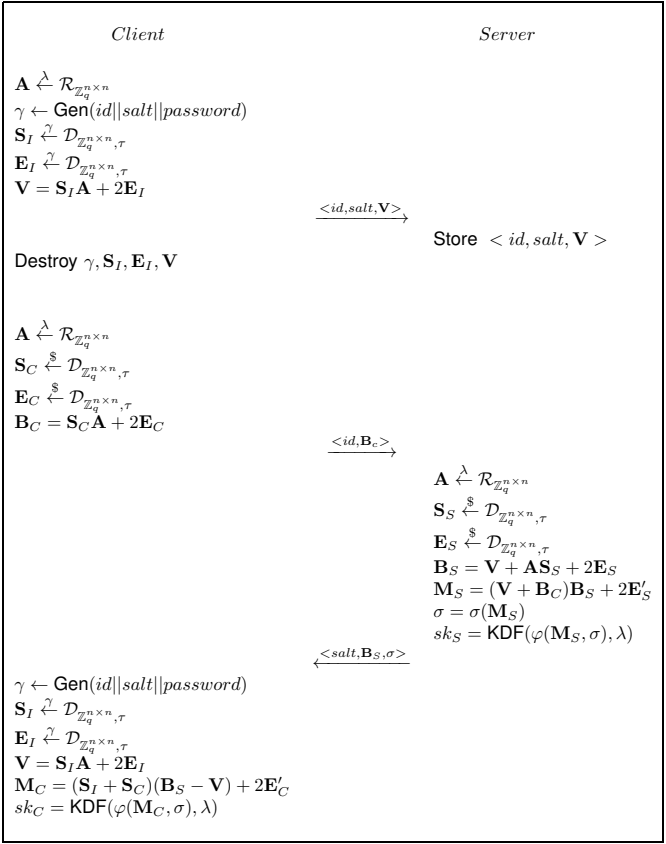

\begin{center}
	\resizebox{1\columnwidth}{!}{
		\begin{threeparttable}
			\begin{tabular}{|lcl|}
				\hline&&\\
				~~~~~~~~~~~~~~~$Client$& &~~~~~~~$Server$\\
				&&\\
				
				$\mathbf{A}\stackrel{\lambda}{\leftarrow} \mathcal{R}_{\mathbb{Z}_{q}^{n \times n}}$&&\\
				
				$\gamma \leftarrow \textsf{Gen}(id||salt||password)$&&\\
				
				$\mathbf{S}_{I}\stackrel{\gamma}{\leftarrow}  \mathcal{D}_{\mathbb{Z}_{q}^{n \times n}, \tau}$&&\\
				
				$\mathbf{E}_{I} \stackrel{\gamma}{\leftarrow} \mathcal{D}_{\mathbb{Z}_{q}^{n \times n}, \tau}$&&\\
				
				$\mathbf{V} =\mathbf{S}_{I}\mathbf{A} +2\mathbf{E}_{I}$&&\\
				
				&$\xrightarrow{<id, salt, \mathbf{V}>}$&\\
				
				&&$\textsf{Store} \ <id, salt, \mathbf{V}>$\\
				
				$\textsf{Destroy} \ \gamma, \mathbf{S}_{I}, \mathbf{E}_{I}, \mathbf{V}$&&\\
				
				&&\\
				&&\\
				
				$\mathbf{A}\stackrel{\lambda}{\leftarrow} \mathcal{R}_{\mathbb{Z}_{q}^{n \times n}}$&&\\

				$\mathbf{S}_{C}\stackrel{\$}{\leftarrow}  \mathcal{D}_{\mathbb{Z}_{q}^{n \times n}, \tau}$&&\\

				$\mathbf{E}_{C} \stackrel{\$}{\leftarrow} \mathcal{D}_{\mathbb{Z}_{q}^{n \times n}, \tau}$&&\\

				$\mathbf{B}_{C} =\mathbf{S}_{C}\mathbf{A} +2\mathbf{E}_{C}$&&\\

				&$\xrightarrow{<id,\mathbf{B}_{c}>}$&\\

				&&$\mathbf{A}\stackrel{\lambda}{\leftarrow} \mathcal{R}_{\mathbb{Z}_{q}^{n \times n}}$\\

				&&$\mathbf{S}_{S}\stackrel{\$}{\leftarrow}  \mathcal{D}_{\mathbb{Z}_{q}^{n \times n}, \tau}$\\

				&&$\mathbf{E}_{S} \stackrel{\$}{\leftarrow} \mathcal{D}_{\mathbb{Z}_{q}^{n \times n}, \tau}$\\

				&&$\mathbf{B}_{S} =\mathbf{V}+\mathbf{A}\mathbf{S}_{S} +2\mathbf{E}_{S}$\\

				&&$\mathbf{M}_{S} =(\mathbf{V}+\mathbf{B}_{C})\mathbf{B}_{S}+2\mathbf{E}_{S}^{\prime}$\\
				
				&&$\sigma=\sigma(\mathbf{M}_{S})$\\
				
				&&$sk_{S}=\textsf{KDF}(\varphi(\mathbf{M}_{S}, \sigma),\lambda)$\\

				&$\xleftarrow{<salt, \mathbf{B}_{S}, \sigma>}$&\\
				
				$\gamma \leftarrow \textsf{Gen}(id||salt||password)$&&\\
				
				$\mathbf{S}_{I}\stackrel{\gamma}{\leftarrow}  \mathcal{D}_{\mathbb{Z}_{q}^{n \times n}, \tau}$&&\\
				
				$\mathbf{E}_{I} \stackrel{\gamma}{\leftarrow} \mathcal{D}_{\mathbb{Z}_{q}^{n \times n}, \tau}$
				&&\\
				
				$\mathbf{V} =\mathbf{S}_{I}\mathbf{A} +2\mathbf{E}_{I}$&&\\	
				
				$\mathbf{M}_{C} =(\mathbf{S}_{I}+\mathbf{S}_{C})(\mathbf{B}_{S}-\mathbf{V})+2\mathbf{E}_{C}^{\prime}$&&\\
				
				$sk_{C}=\textsf{KDF}(\varphi(\mathbf{M}_{C}, \sigma),\lambda)$&&\\
				
				&&\\
				\hline
			\end{tabular}
	\end{threeparttable}}
\end{center}
\caption{The \textsf{\textmd{LWE}}-based \textsf{\textmd{SRP}} protocol.}
\label{lwesrp}
\end{figure}
\subsection{Correctness}
\subsubsection{Calculations of key materials} To ensure that both parties generate the same session key, an important part of the protocol is to calculate the key materials. The specific calculations are as follows:
	\begin{itemize}
		\item $Client$: \begin{footnotesize}\begin{equation*}
			\begin{aligned}
					\mathbf{M}_{C} &=(\mathbf{S}_{I}+\mathbf{S}_{C})(\mathbf{B}_{S}-\mathbf{V})+2\mathbf{E}_{C}^{\prime}\\
					&=(\mathbf{S}_{I}+\mathbf{S}_{C})(\mathbf{A}\mathbf{S}_{S} +2\mathbf{E}_{S})+2\mathbf{E}_{C}^{\prime}\\
					&=\mathbf{S}_{I}\mathbf{A}\mathbf{S}_{S}+\mathbf{S}_{C}\mathbf{A}\mathbf{S}_{S}+2(\mathbf{S}_{I}\mathbf{E}_{S}+\mathbf{S}_{C}\mathbf{E}_{S}+\mathbf{E}_{C}^{\prime})\\
					&\stackrel{\varphi}{\Rightarrow}\mathbf{S}_{I}\mathbf{A}\mathbf{S}_{S}+\mathbf{S}_{C}\mathbf{A}\mathbf{S}_{S},\\
			\end{aligned}
		\end{equation*}\end{footnotesize}
			\item $Server$: \begin{footnotesize}\begin{equation*}
		\begin{aligned}
			\mathbf{M}_{s} &=(\mathbf{V}+\mathbf{B}_{C})\mathbf{S}_{S}+2\mathbf{E}_{S}^{\prime}\\
			&=(\mathbf{S}_{I}\mathbf{A} +2\mathbf{E}_{I}+\mathbf{B}_{C})\mathbf{S}_{S}+2\mathbf{E}_{S}^{\prime}\\
			&=\mathbf{S}_{I}\mathbf{A}\mathbf{S}_{S}+(\mathbf{S}_{C}\mathbf{A} +2\mathbf{E}_{C})\mathbf{S}_{S}+2\mathbf{E}_{i}\mathbf{S}_{S}+2\mathbf{E}_{S}^{\prime}\\
			&=\mathbf{S}_{I}\mathbf{A}\mathbf{S}_{S}+\mathbf{S}_{C}\mathbf{A}\mathbf{S}_{S}+2(\mathbf{E}_{C}\mathbf{S}_{S}+\mathbf{E}_{I}\mathbf{S}_{S}+\mathbf{E}_{S}^{\prime})\\
			&\stackrel{\varphi}{\Rightarrow}\mathbf{S}_{I}\mathbf{A}\mathbf{S}_{S}+\mathbf{S}_{C}\mathbf{A}\mathbf{S}_{S}.\\
		\end{aligned}
	\end{equation*}\end{footnotesize}
	
	\end{itemize}

And now, it is necessary to prove the correctness of the calculations.

\subsubsection{Correctness guaranteed via $\sigma(\cdot)$ and $\varphi(\cdot)$}
\begin{lemma}{\rm (Correctness of extractor function \cite{ding-Asimpleprovably})}\label{lem1}
	Let $q>8$ be an odd integer, then the function $\varphi(\cdot)$ defined above guarantees the correctness concerning $\sigma(\cdot)$ with the error tolerance $\delta \leqslant \frac{q}{4}-2$.
\end{lemma}
\begin{proof}
	Let $2\varepsilon=\delta \leqslant \frac{q}{4}-2$. For any $x, y \in \mathbb{Z}_{q}$ such that $x-y=2 \varepsilon$  and $|2 \varepsilon| \leqslant \frac{q}{4}-2$. Let $\sigma \leftarrow \sigma(\cdot)$, due to the definition of $\sigma$: $\left|y+\sigma \cdot \frac{q-1}{2} \bmod q\right| \leqslant \frac{q}{4}+1$ for any input used in $\sigma(\cdot)$ to generate $\sigma$. And then:
		\begin{footnotesize}
			\begin{equation*}
				\begin{split}
					\left|\left(y+\sigma \cdot \frac{q-1}{2}\right) \bmod q+2 \varepsilon\right| &\leqslant \frac{q}{4}+1+|2 \varepsilon|\\
					&\leqslant \frac{q-1}{2},\\	
				\end{split}
			\end{equation*}
		\end{footnotesize}
	
	
	\begin{footnotesize}
		\begin{equation*}
			\begin{split}
				\Rightarrow
				x+\sigma \cdot \frac{q-1}{2} \bmod q&=y+\sigma \cdot \frac{q-1}{2}+2 \varepsilon \bmod q\\
				&=\left(y+\sigma \cdot \frac{q-1}{2}\right) \bmod q+2 \varepsilon,
			\end{split}
		\end{equation*}	
	\end{footnotesize}
	
	
	\begin{footnotesize}
		\begin{equation*}
			\begin{split}
				\Rightarrow
				\varphi(x, \sigma)&=\left(x+\sigma \cdot \frac{q-1}{2} \bmod q\right) \bmod 2\\
				&\xlongequal{\textgreater 1-negl(\lambda)}\left(y+\sigma \cdot \frac{q-1}{2} \bmod q\right) \bmod 2\\
				&=\varphi(y, \sigma).\\
			\end{split}
		\end{equation*}	
	\end{footnotesize}
\end{proof}

\begin{theorem}
	The correctness is guaranteed in the form of:
	 $\varphi(\mu_{i,j}^{0})=\varphi(\mu_{i,j}^{1})$ with a overwhelming probability when the error tolerance $\delta = 12(\tau \sqrt{n})^{2}\leqslant\frac{q}{4}-2$. And then both parties can derive the same session key via $\varphi(\cdot)$, where $\mu_{i,j}^{k}\in \left\{\mu_{i,j}^{k}:[\mu_{i,j}^{k}]_{n \times n} \in \mathbb{Z}_{q}^{n \times n},  k \in \{0,1\}\right\}$ 
\label{th2}		
\end{theorem}
\begin{proof}
	Assuming $\mathbf{A}_{n \times n}=[\boldsymbol{a}_{1},\boldsymbol{a}_{2},\ldots\boldsymbol{a}_{n}]$, $q$ is the modulus of $\mathbb{Z}_{q}^{n \times n}$ and the parameter of Gaussian distribution $\tau = \alpha q$.
	Let $\left\|\mathbf{A}\right\|:=\max\limits_{1 \leqslant i,j \leq n}\left\|a_{i,j}\right\|=\max\limits_{1 \leqslant i \leqslant n}\left\|\boldsymbol{a}_{i}\right\|_{\infty}$.
	\begin{footnotesize}
		\begin{equation*}
		\begin{aligned}
			\begin{split}
				\left\|\mathbf{M}_{C}-\mathbf{M}_{S}\right\|
				=&\ \|\mathbf{S}_{I}\mathbf{A}\mathbf{S}_{S}+\mathbf{S}_{C}\mathbf{A}\mathbf{S}_{S}+2(\mathbf{S}_{I}\mathbf{E}_{S}+\mathbf{S}_{C}\mathbf{E}_{S}+\mathbf{E}_{C}^{\prime})-\\
				&\ \mathbf{S}_{I}\mathbf{A}\mathbf{S}_{S}+\mathbf{S}_{C}\mathbf{A}\mathbf{S}_{S}+2(\mathbf{E}_{C}\mathbf{S}_{S}+\mathbf{E}_{I}\mathbf{S}_{S}+\mathbf{E}_{S}^{\prime})\|\\
				=&\ 2\left\|\mathbf{S}_{I}\mathbf{E}_{S}+\mathbf{S}_{C}\mathbf{E}_{S}+\mathbf{E}_{C}^{\prime}-\mathbf{E}_{C}\mathbf{S}_{S}-\mathbf{E}_{I}\mathbf{S}_{S}-\mathbf{E}_{S}^{\prime}\right\|\\
				=&\ \left\|8\mathbf{S}_{*}\mathbf{E}_{*}+4\mathbf{E}_{*}\right\|\\
				\leqslant&\ 12\left\|\mathbf{S}_{*}\mathbf{E}_{*}\right\|\\
				\leqslant&\ 12(\alpha q \sqrt{n})^{2}\\
				\leqslant&\ \frac{q}{4}-2\\
			\end{split}
		\end{aligned}
	\end{equation*}
\end{footnotesize}
That is there exists:
\begin{center}
\begin{footnotesize}
	$\operatorname{Pr}\left[\left\|\mathbf{M}_{c}-\mathbf{M}_{s}\right\| \geqslant \frac{q}{4}-2, [\mu_{i,j}^{k}]_{n \times n} \leftarrow \mathcal{D}_{\mathbb{Z}_{q}^{n \times n}, \tau} \right]<negl(\lambda)$.
\end{footnotesize}
\end{center}
According to Lemma \ref{lem1}, both parties obtain the equivalent key material via $\varphi(\cdot)$ with overwhelming probability.
\end{proof}

\subsection{Security}

The \textsf{\textmd{SRP}} protocol is widely used (e.g., \cite{srp-app-tls-Usingthesecure}, \cite{dhass-Ananalysisof}, etc.), and its security has been recognized in application. However, there is no strict formal security proof for it currently. In this section, we give strict security proof based on the scheme proposed by Bellare et al. [18] (\textsf{\textmd{BPR}} model) for the proposed \textsf{\textmd{LWE}}-based \textsf{\textmd{SRP}} protocol. Besides, we also present security analyses of our design considering quantum attacks complying with the original paper of \textsf{\textmd{SRP}} \cite{srp0}.

\subsubsection{Security model}: Firstly, we introduce the \textsf{\textmd{BPR}} security
model. (Please refer to the original paper [18] for the complete description.)

\begin{definition}{\rm (\textsf{BPR} security model \cite{BPR-Authenticatedkeyexchange})}
	 Instance $i$ of principal $U$ is called ``an oracle", where each principal $U \in I D$, denoted as $\Pi_{U}^{i}$. Queries that adversary $\mathcal{A}$ can make are as follows:
	\begin{itemize}
		\item {\rm \textsf{Send}$(U, i, M)$:} Sends message $M$ to oracle $\Pi_U^i$. This simulates a protocol instance generated by an adversary.
		\item {\rm \textsf{Reveal}$(U, i)$:} Returns the correct session key to the adversary from oracle $\Pi_U^i$. This simulates the loss of a session key.
		\item {\rm \textsf{Corrupt}$(U, p w)$:} The adversary obtains $pw_U$ and the states of all instances of $U$. This simulates a disastrous damage of the server.
		\item {\rm \textsf{Execute}$(C, i, S, j)$:} If there is no protocol being executed between the client and the server, instantiate a protocol with the input of this query.
		\item {\rm \textsf{Test}$(U, i)$:} Returns the session key holding by $\Pi^U_i$ or a random key according to the value of $b \leftarrow \{0, 1\}$. This is used to evaluate whether the attacker has successfully obtained the password or not.
	\end{itemize}
\end{definition}

\subsubsection{Game-based proof}
\begin{theorem} If the decisional-{\rm \textsf{LWE}} problem is hard, $\sigma(\cdot)$ and $\varphi(\cdot)$ are correct with overwhelming probability, {\rm \textsf{Gen}} is a secure hash function and {\rm \textsf{KDF}} is a secure key derivation function, then the proposed {\rm \textsf{LWE}}-based {\rm \textsf{SRP}} protocol is secure.
	\label{th3}
\end{theorem}
\begin{proof}The following security proof is based on the \textsf{BPR} model.
	\begin{itemize}
		\item \textsf{Game}$_{0}$: The real security game. Adversary $\mathcal{A}$ outputs a guess bit $b'$ for $b$ at the end of \textsf{Game}$_{0}$. Here $\mathbf{Adv}[\mathcal{A}_{0}] := \left| \Pr [b' = b] - \frac{1}{2} \right|$ denotes the advantage.
		

		\item {\rm \textsf{Game}$_{1}$ \textsf{:}} This game is based on \textsf{Game}$_{0}$ with the replacement of salt. The adversary $\mathcal{A}_1$ sends the query \textsf{Execute}$(C, i, S, j; \text{salt})$ to \textsf{Client}. The random seed $\gamma$ then is generated. Due to the indistinguishability of salt and the security of hash functions, $\mathcal{A}_1$ can distinguish between random seeds with the advantage $\mathbf{Adv}[\mathcal{A}_1] \leq \mathbf{Adv}[\mathcal{A}_0] + \epsilon_1(\kappa)$, where $\epsilon_1(\kappa)$ is negligible for adversary $\mathcal{A}_1$, then we get $\left| \Pr[\textsf{Game}_1] - \Pr[\textsf{Game}_0] \right| \leq \text{negl}(\kappa)$;
		
		\item {\rm \textsf{Game}$_{2}$ \textsf{:}}
		This game is based on \textsf{Game}$_{1}$ with the replacement of the initial secret $S_I$. The adversary $\mathcal{A}_2$ sends the query \textsf{Execute}$(C, i, S, j; S_I)$ to \textsf{Client}. $\mathcal{A}_2$ can distinguish between different LWE instances with the advantage $\mathbf{Adv}[\mathcal{A}_2] \leq \mathbf{Adv}[\mathcal{A}_1] + \epsilon_2(\kappa)$, where $\epsilon_2(\kappa)$ is negligible for adversary $\mathcal{A}_2$, then we get $\left| \Pr[\textsf{Game}_2] - \Pr[\textsf{Game}_1] \right| \leq \text{negl}(\kappa)$;
		
		\item {\rm \textsf{Game}$_{3}$ \textsf{:}}
		This game is based on \textsf{Game}$_{2}$ with the replacement of the initial noise $E_I$. The adversary $\mathcal{A}_3$ sends the query \textsf{Execute}$(C, i, S, j; E_I)$ to \textsf{Client}. $\mathcal{A}_3$ can distinguish between different LWE instances with the advantage $\mathbf{Adv}[\mathcal{A}_3] \leq \mathbf{Adv}[\mathcal{A}_2] + \epsilon_3(\kappa)$, where $\epsilon_3(\kappa)$ is negligible for adversary $\mathcal{A}_3$, then we get $\left| \Pr[\textsf{Game}_3] - \Pr[\textsf{Game}_2] \right| \leq \text{negl}(\kappa)$;
		
		\item {\rm \textsf{Game}$_{4}$ \textsf{:}}
		This game is based on \textsf{Game}$_{3}$ with the replacement of the verifier $V$. The adversary $\mathcal{A}_4$ sends the query \textsf{Send}$(S, j; V)$ to \textsf{Server}. $\mathcal{A}_4$ can distinguish between different LWE instances with the advantage $\mathbf{Adv}[\mathcal{A}_4] \leq \mathbf{Adv}[\mathcal{A}_3] + \epsilon_4(\kappa)$, where $\epsilon_4(\kappa)$ is negligible for adversary $\mathcal{A}_4$, then we get $\left| \Pr[\textsf{Game}_4] - \Pr[\textsf{Game}_3] \right| \leq \text{negl}(\kappa)$;
		
		\item {\rm \textsf{Game}$_{5}$ \textsf{:}}
		This game is based on \textsf{Game}$_{4}$ with the replacement of the noise $B_C$. The adversary $\mathcal{A}_5$ sends the query \textsf{Send}$(S, j; B_C)$ to \textsf{Server}. $\mathcal{A}_5$ can distinguish between different LWE instances with the advantage $\mathbf{Adv}[\mathcal{A}_5] \leq \mathbf{Adv}[\mathcal{A}_4] + \epsilon_5(\kappa)$, where $\epsilon_5(\kappa)$ is negligible for adversary $\mathcal{A}_5$, then we get $\left| \Pr[\textsf{Game}_5] - \Pr[\textsf{Game}_4] \right| \leq \text{negl}(\kappa)$;
		
		\item {\rm \textsf{Game}$_{6}$ \textsf{:}}
		This game is based on \textsf{Game}$_{5}$ with the replacement of salt. The adversary $\mathcal{A}_6$ sends the query \textsf{Send}$(C, i; \text{salt})$ to \textsf{Server}. Due to the indistinguishability of salt and the security of hash functions, $\mathcal{A}_4$ can distinguish between different random seeds with the advantage $\mathbf{Adv}[\mathcal{A}_6] \leq \mathbf{Adv}[\mathcal{A}_5] + \epsilon_6(\kappa)$, where $\epsilon_6(\kappa)$ is negligible for adversary $\mathcal{A}_6$, then we get $\left| \Pr[\textsf{Game}_6] - \Pr[\textsf{Game}_5] \right| \leq \text{negl}(\kappa)$;
		
		\item {\rm \textsf{Game}$_{7}$ \textsf{:}}
		This game is based on \textsf{Game}$_{6}$ with the replacement of the noise $B_S$. The adversary $\mathcal{A}_7$ sends the query \textsf{Send}$(C, i; B_S)$ to \textsf{Client}. $\mathcal{A}_6$ can distinguish between different LWE instances with the advantage $\mathbf{Adv}[\mathcal{A}_7] \leq \mathbf{Adv}[\mathcal{A}_6] + \epsilon_7(\kappa)$, where $\epsilon_7(\kappa)$ is negligible for adversary $\mathcal{A}_7$, then we get $\left| \Pr[\textsf{Game}_7] - \Pr[\textsf{Game}_6] \right| \leq \text{negl}(\kappa)$;
		
		\item {\rm \textsf{Game}$_{8}$ \textsf{:}}
		This game is based on \textsf{Game}$_{7}$ with the replacement of $\sigma$. The adversary $\mathcal{A}_8$ sends the query \textsf{Send}$(C, i; \sigma)$ to \textsf{Client} and $\sigma$ is sampled randomly from $\{0, 1\}^{n \times n}$. Due to the randomness of $\sigma$ and indistinguishability of the LWE instance $M_C \phi \Rightarrow S_{IAS}S + S_{CAS}S$ generated after calculating the key materials, $\mathcal{A}_8$ can distinguish between different LWE instances with the advantage $\mathbf{Adv}[\mathcal{A}_8] \leq \mathbf{Adv}[\mathcal{A}_7] + \epsilon_8(\kappa)$, where $\epsilon_8(\kappa)$ is negligible for adversary $\mathcal{A}_8$, then we get $\left| \Pr[\textsf{Game}_8] - \Pr[\textsf{Game}_7] \right| \leq \text{negl}(\kappa)$.
	\end{itemize}

In \textsf{Game}$_{8}$, all the parameters have been changed randomly. In this case, an adversary cannot distinguish between different LWE instances due to Theorem 1. Hence, the probability of adversary $\mathcal{A}_8$ winning \textsf{Game}$_{8}$ is $\Pr[\textsf{Game}_8] = \frac{1}{2}$, i.e., $\mathbf{Adv}[\mathcal{A}_8] = \frac{1}{2}$. Finally we get $\mathbf{Adv}[\mathcal{A}_0] = \left| \Pr[b' = b] - \frac{1}{2} \right| \leq \text{negl}(\kappa)$. The process shows an adversary cannot discover the password of the protocol under the given security model.
\end{proof}
\subsubsection{Analyses of the execution processes}
Based on the security descriptions of the original \textsf{\textmd{SRP}} \cite{srp0}, we give detailed security analyses of the execution process of the proposed \textsf{LWE}-based \textsf{SRP} protocol.

\begin{itemize}
	\item \textbf{Resistance to quantum attacks.} The scheme is designed from the \textsf{LWE} problem. An adversary can solve this hard problem if the security of the protocol can be broken. However, at present, no public algorithm has been designed to solve the hard problem in polynomial-time. Therefore, our scheme overcomes the disadvantage that the original protocol cannot resist quantum attacks, and has high security.

	\item \textbf{No recovery of the password or a session key.} Firstly, the probability of recovering the password or a session key is negligible. This is guaranteed via the indistinguishability of the decisional-\textsf{LWE} problem; Second, the information associated with the password is perfectly hidden. The client first needs to generate the random seed with the participation of the password and use this seed to construct the initial verifier, an \textsf{LWE} instance. This means that an attacker must break the security of hash functions and solve the \textsf{LWE} problem if he intends to recover the password. However, both are impossible.
	
	\item \textbf{Security without dependence on the verifier.} Even if the verifier is leaked, the adversary still cannot obtain the password, and the most he can do is to block the protocol from continuing to run. Specifically, the verifier is essentially an \textsf{LWE} instance, and according to the protocol, the adversary must know the password and generate the corresponding secret value and noise value to generate the session key after computing the key material. Therefore, the adversary cannot imitate either party in the protocol with a stolen verifier.
	
	\item \textbf{Independence among session keys} When the adversary obtains a session key via some method, he cannot infer the previous key or the future key. In other words, considering the worst case, when the server is attacked by an attacker, the protocol can also ensure the security of historical session keys. $Client$ and $Server$ needs to randomly generate secrets for different sessions. Therefore, the session key obtained by each key exchange process depends on the independent \textsf{LWE} problem instances. Due to the difficulty of the decisional-\textsf{LWE} problem, the security of different session keys is guaranteed. Even a more extreme situation can be considered: The client's password itself is compromised additionally. However, the intruder still cannot know the previous session key and decrypt the previous ciphertext in this case due to similar reasons.
	
\end{itemize}

The above analyses cover various situations involving verifier, password, and session key. Combined with strict security proofs, it is shown that the protocol proposed in this paper guarantees security.

\section{Conclusions}\label{Conclusions}

In this paper, we propose an \textsf{LWE}-based SRP protocol based on the conception of the extended multi-instance \textsf{LWE} problem. It mainly has the following advantages: 1). It only depends on the basic \textsf{LWE} construction and is resistant to quantum attacks; 2). The password is not exposed during the process via a verifier; 3). Even if the verifier is leaked, an attacker cannot recover the password from it or get a session key by impersonating the client. We provide strict proof and detailed analyses of the correctness and security of the protocol. The protocol can generate a session shared keys after running successfully and it is resistant to various types of attacks in different scenarios. 

Our design can be optimized from the following aspects:
1). Construction of leaner key material. Although we constructed the correct and secure key material from the most
intuitive and efficient point of view, more sophisticated and
constructive methods still deserve further exploration. 2). More
efficient calculation method. In the calculation process of the
key material, we adopt a relatively simple calculation and
simplification method. Therefore, designing a more efficient
computing process can be considered to improve the protocol.
more parties. 3). Multi parities. The protocol in this paper is
constructed in the scene of a client and a server. It can also
be reconstructed to build an extended identity-based scheme
running among multi parities based on this work. And we
leave these works for future research.

\section*{ACKNOWLEDGMENT}
This work was funded by R\&D Program of Beijing Municipal Education Commission with grant number 110052971921/021.

\bibliographystyle{ieeetr}
\bibliography{references}

\begin{thebibliography}{10}

\bibitem{srp0}
T.~D. Wu {\em et~al.}, ``The secure remote password protocol.,'' in {\em NDSS},
  vol.~98, pp.~97--111, Citeseer, 1998.

\bibitem{z1zhao2021novel}
L.~Zhao, M.~B. Saif, A.~Hawbani, G.~Min, S.~Peng, and N.~Lin, ``A novel
  improved artificial bee colony and blockchain-based secure clustering routing
  scheme for fanet,'' {\em China Communications}, vol.~18, no.~7, pp.~103--116,
  2021.

\bibitem{z2peng2021secure}
S.~Peng, L.~Zhao, A.~Y. Al-Dubai, A.~Y. Zomaya, J.~Hu, G.~Min, and Q.~Wang,
  ``Secure lightweight stream data outsourcing for internet of things,'' {\em
  IEEE Internet of Things Journal}, vol.~8, no.~13, pp.~10815--10829, 2021.

\bibitem{shor1999polynomial}
P.~W. Shor, ``Polynomial-time algorithms for prime factorization and discrete
  logarithms on a quantum computer,'' {\em SIAM review}, vol.~41, no.~2,
  pp.~303--332, 1999.

\bibitem{Introduction}
D.~J. Bernstein, ``Introduction to post-quantum cryptography,'' in {\em
  Post-quantum cryptography}, pp.~1--14, Springer, 2009.

\bibitem{Apublic-key}
M.~Ajtai and C.~Dwork, ``A public-key cryptosystem with worst-case/average-case
  equivalence,'' in {\em Proceedings of the twenty-ninth annual ACM symposium
  on Theory of computing}, pp.~284--293, 1997.

\bibitem{svp}
M.~Ajtai, ``The shortest vector problem in l2 is np-hard for randomized
  reductions,'' in {\em Proceedings of the thirtieth annual ACM symposium on
  Theory of computing}, pp.~10--19, 1998.

\bibitem{Adecadeof-decade}
C.~Peikert {\em et~al.}, ``A decade of lattice cryptography,'' {\em Foundations
  and trends{\textregistered} in theoretical computer science}, vol.~10, no.~4,
  pp.~283--424, 2016.

\bibitem{cvp}
D.~Micciancio and S.~Goldwasser, ``Closest vector problem,'' in {\em Complexity
  of Lattice Problems}, pp.~45--68, Springer, 2002.

\bibitem{Approximatingcsvphard}
O.~Goldreich, D.~Micciancio, S.~Safra, and J.-P. Seifert, ``Approximating
  shortest lattice vectors is not harder than approximating closest lattice
  vectors,'' {\em Information Processing Letters}, vol.~71, no.~2, pp.~55--61,
  1999.

\bibitem{mlwe-Algebraically}
C.~Peikert and Z.~Pepin, ``Algebraically structured lwe, revisited,'' in {\em
  Theory of Cryptography Conference}, pp.~1--23, Springer, 2019.

\bibitem{ding-Asimpleprovably}
J.~Ding, X.~Xie, and X.~Lin, ``A simple provably secure key exchange scheme
  based on the learning with errors problem,'' {\em Cryptology ePrint Archive},
  2012.

\bibitem{rlwesrp}
X.~Gao, J.~Ding, J.~Liu, and L.~Li, ``Post-quantum secure remote password
  protocol from rlwe problem,'' in {\em Information Security and Cryptology:
  13th International Conference, Inscrypt 2017, Xi'an, China, November 3--5,
  2017, Revised Selected Papers 13}, pp.~99--116, Springer, 2018.

\bibitem{regev2009lattices-Onlattices}
O.~Regev, ``On lattices, learning with errors, random linear codes, and
  cryptography,'' {\em Journal of the ACM (JACM)}, vol.~56, no.~6, pp.~1--40,
  2009.

\bibitem{dh-Newdirections}
W.~Diffie and M.~E. Hellman, ``New directions in cryptography,'' in {\em Secure
  communications and asymmetric cryptosystems}, pp.~143--180, Routledge, 2019.

\bibitem{lwe-ana-Robustness}
S.~Goldwasser, Y.~T. Kalai, C.~Peikert, and V.~Vaikuntanathan, ``Robustness of
  the learning with errors assumption,'' 2010.

\bibitem{srp-app-tls-Usingthesecure}
D.~Taylor, T.~Wu, N.~Mavrogiannopoulos, and T.~Perrin, ``Using the secure
  remote password (srp) protocol for tls authentication,'' tech. rep., 2007.

\bibitem{dhass-Ananalysisof}
M.~Abdalla, M.~Bellare, and P.~Rogaway, ``The oracle diffie-hellman assumptions
  and an analysis of dhies,'' in {\em Cryptographers’ Track at the RSA
  Conference}, pp.~143--158, Springer, 2001.

\bibitem{BPR-Authenticatedkeyexchange}
M.~Bellare, D.~Pointcheval, and P.~Rogaway, ``Authenticated key exchange secure
  against dictionary attacks,'' in {\em International conference on the theory
  and applications of cryptographic techniques}, pp.~139--155, Springer, 2000.

\end{thebibliography}

\end{document}